\documentclass[12pt,reqno]{article}

\usepackage[usenames]{color}
\usepackage{amssymb}
\usepackage{amsmath}
\usepackage{amsthm}
\usepackage{amsfonts}
\usepackage{amscd}
\usepackage{graphicx}
\usepackage{xcolor}
\usepackage{float}

\usepackage[colorlinks=true,
linkcolor=webgreen,
filecolor=webbrown,
citecolor=webgreen]{hyperref}

\definecolor{webgreen}{rgb}{0,.5,0}
\definecolor{webbrown}{rgb}{.6,0,0}
\usepackage{fullpage}

\usepackage{graphics}
\usepackage{latexsym}
\usepackage{epsf}
\usepackage{breakurl}

\theoremstyle{plain}
\newtheorem{theorem}{Theorem}

\newtheorem{lemma}[theorem]{Lemma}
\newtheorem{proposition}[theorem]{Proposition}

\theoremstyle{definition}

\theoremstyle{remark}

\newcommand{\N}{\mathbb{N}}
\newcommand{\Z}{\mathbb{Z}}

\title{Solution to Bucher's density problem\\
for context-free languages}
\author{Rastko Maslic\\
Belgrade\\
Serbia\\
\href{mailto:rastko.maslic@gmail.com}{\tt rastko.maslic@gmail.com}\\
\and Jeffrey Shallit\\
School of Computer Science\\
University of Waterloo\\
Waterloo, ON  N2L 3G1\\
Canada\\
\href{mailto:shallit@uwaterloo.ca}{\tt shallit@uwaterloo.ca}}
\date{7 September 2026}

\begin{document}
\maketitle

\begin{abstract}
\normalsize
In 1980 Bucher asked whether, given context-free languages
\(L\subseteq U\) with \(U\setminus L\) infinite, there must be a
context-free language \(K\) between them for which both
\(K\setminus L\) and \(U\setminus K\) are infinite. We give a
negative answer. 

We first construct an infinite language \(D\) with
context-free complement such that, for every regular language \(R\),
either \(D\cap R\) or \(D\setminus R\) is finite. The words of \(D\)
encode computations of factorials; repetition of letters ensures that each finite automaton either accepts all but finitely many words of \(D\) or rejects all but finitely many words of \(D\),
while a one-counter automaton recognizes errors in the encodings.
We then construct \(L\) and \(U\) from the complement of \(D\).
A grammar argument shows that any context-free intermediate language $K$ would
divide \(D\) in the same way as some regular language. This proves
the required impossibility. Both \(L\) and \(U\) can be taken over a
binary alphabet.
\end{abstract}

\section{Bucher's problem and the main theorem}\label{sec:intro}

\subsection{The original question}

Bucher \cite{ref:bucher} asked the following question in 1980: given a family of languages \(\mathcal L\) and
\(L, U\in\mathcal L\) such that
\[
 L\subseteq U,\qquad |U\setminus L|=\infty,
\]
must there be \(K\in\mathcal L\) with
\begin{equation}\label{eq:bucher}
 L\subseteq K\subseteq U,\qquad
 |U\setminus K|= \infty, \qquad |K\setminus L|=\infty?
\end{equation}

Bucher gave affirmative answers for the regular, recursive, and
deterministic context-sensitive languages, and singled out the context-free languages as an unresolved case.
The motivation came from
density questions in grammatical similarity \cite{ref:mswforms,ref:mswspaces}.

We solve Bucher's problem by proving that the answer for context-free languages is negative:
\begin{theorem}\label{thm:main}
There is a context-free language \(L\subseteq U\subseteq\{0,1\}^*\)
such that \(U\setminus L\) is infinite, but every context-free
language \(K\) with \(L\subseteq K\subseteq U\) satisfies
\[
 |K\setminus L|<\infty
 \quad\hbox{or}\quad
 |U\setminus K|<\infty.
\]
\end{theorem}

\subsection{Outline of the proof}

All alphabets in the paper are finite. For an alphabet \(\Sigma\), we let
\(\Sigma^*\) denote the set of all finite words over \(\Sigma\),
including the empty word \(\varepsilon\). We write \(w^R\) for the
reversal of a word \(w\), and \(A^R=\{w^R:w\in A\}\) for the
reversal of a language. Also, \(\N=\{0,1,2,\ldots\}\), and \(\Z\)
denotes the integers.

The construction starts with an infinite language \(D\subseteq\Sigma^*\)
having the following two properties:
\begin{enumerate}
\item Its complement \(A=\Sigma^*\setminus D\) is context-free.
\item For every regular \(R\subseteq\Sigma^*\), at least one of
      \(D\cap R\) and \(D\setminus R\) is finite.
\end{enumerate}
The second property says that a finite automaton cannot accept
infinitely many words of \(D\) and reject infinitely many others.

Choose a symbol \(\#\) that is not in \(\Sigma\). It serves as a
\emph{separator}: in the words below it occurs exactly once, between
two words over \(\Sigma\). Define
\begin{equation}\label{eq:endpoints}
 \begin{split}
 L&=A\#A^R=\{x\#y^R:x,y\in A\},\\
 U&=L\cup\{w\#w^R:w\in\Sigma^*\}.
 \end{split}
\end{equation}
The words \(x\) and \(y\) in the definition of \(L\) are chosen
independently. Both \(L\) and \(U\) are context-free, and
\[
 U\setminus L=\{w\#w^R:w\in D\}.
\]
Section~\ref{sec:grammar} proves that, for any context-free
\(K\subseteq U\),  when \(w\in D\), membership of \(w\#w^R\) in \(K\)
agrees with membership of \(w\) in a suitable regular language.
The second property of \(D\) then gives the conclusion of the theorem.

Sections~\ref{sec:counter}--\ref{sec:factorial} construct the language \(D\).
For each \(n\geq4\), we encode a computation of \(n! = 1 \cdot 2 \cdots n\) by a word
\(w_n\). We extend the computation to exactly \(n!\) steps by
appending steps in which nothing changes, and use multiples of \(n!\) as
the lengths of all the repeated-letter portions of the encoding.
We prove that errors in such words can be recognized by a
one-counter automaton, while every finite automaton treats all $w_n$ the same, if $n$ is sufficiently large.
Section~\ref{sec:mainproof}
combines the construction with the grammar argument and then
encodes the resulting languages over two letters.

\subsection{Earlier work}

Domaratzki discussed Bucher's problem in Sections~2.4 and~5.4 of
his thesis \cite{ref:dom}, in connection with his joint work
with Shallit and Yu \cite{ref:dsy} on regular and context-free
supersets of languages.
In Theorems~4.3.1--4.3.3 of the thesis, he used palindromes to
relate questions about regular supersets to corresponding
questions about context-free supersets.
Horv\'ath, Karhum\"aki, and Kleijn \cite{ref:hkk} had proved a
structural characterization of context-free languages consisting of
palindromes. D\"om\"osi, Fazekas, and Ito
\cite[Thm.~13,]{ref:dfi} later gave another proof.
These results provided precedents for deriving restrictions on a
context-free grammar from the requirement that certain words have
matching reversed halves. Our argument also allows words whose
halves are different, namely the words in \(A\#A^R\).

Yamakami and Kato \cite{ref:yk} studied \emph{regular dissection}:
a regular language \(R\) dissects a language \(E\) when both
\(E\cap R\) and \(E\setminus R\) are infinite. In particular, they
used the unary language \(\{a^{n!}:n\geq1\}\) to show that an
infinite language need not admit a regular dissection. Their
Section~6 asked whether every infinite language with context-free
complement admitted such a dissection. The same question appeared
in the DCFS~2015 problem list \cite[Section~5,
Question~1]{ref:dcfs}. The language \(D\) needed for the present proof
satisfies exactly the contrary property.

The second author restated Bucher's question as Open Problem~6 in
\cite[Slide 28]{ref:shallit}. Sin'ya
\cite[Cor.~1]{ref:sinya} proved an affirmative result when
both given languages are unambiguous context-free languages.

The use of automata to recognize invalid computations appeared, for example, in
Hartmanis \cite{ref:hartmanis} and in the exposition of Hopcroft
and Ullman \cite[Section 8.6]{ref:hu}. Hoogeboom
\cite[Thm.~9]{ref:hoogeboom} also presented this method.
An automaton guesses an error in a word that purports to describe
a computation and checks that error. Section~\ref{sec:encoding}
gives the full argument for the encoding used here, including the
repetition of letters required by the finite-automaton argument.

\section{Extracting a regular language from a grammar}\label{sec:grammar}

For \(E\subseteq\Sigma^*\), write
\[
 \Delta(E)=\{w\#w^R:w\in E\}.
\]
Fix any language \(A\subseteq\Sigma^*\), and put
\[
 D=\Sigma^*\setminus A,\qquad U_A=A\#A^R\cup\Delta(\Sigma^*).
\]
The next lemma says that, on words indexed by \(D\), a
context-free sublanguage of \(U_A\) makes the analogous membership
decisions as a certain regular language over \(\Sigma\).
Here \(A\) need not be context-free.
The lemma extends the marked-palindrome consequence of the characterization of Horv{\'a}th, Karhumäki, and Kleijn \cite{ref:hkk}. Related applications to intermediate languages appeared in Domaratzki's work on minimal covers \cite{ref:dom}. The additional words allowed here require a new argument.

\begin{lemma}\label{lem:grammar}
For every context-free \(K\subseteq U_A\), there is a regular
\(R\subseteq\Sigma^*\) such that
\begin{equation}\label{eq:agreement}
 w\in D\quad\Longrightarrow\quad
 \bigl(w\#w^R\in K\ \Longleftrightarrow\ w\in R\bigr).
\end{equation}
\end{lemma}

\begin{proof}
If \(K=\varnothing\), take \(R=\varnothing\). Otherwise choose a
context-free grammar \(G\) for \(K\), and remove every variable
and production that never occur in a derivation of a word of \(K\).
Let \(S\) be the start symbol. We use \(\Rightarrow^*\) to denote
zero or more grammar derivation steps.

\paragraph{Classifying variables by occurrences of the separator.}
For every remaining variable \(B\), there are terminal words
\(u,v\) with \(S\Rightarrow^*uBv\), and \(B\) derives at least one
terminal word. Replacing \(B\) by any word that it derives must give
a word of \(K\). Since every word of \(K\) has exactly one occurrence
of \(\#\), all the words derived from \(B\) have the same number
of occurrences of \(\#\). This number is either zero or one.
Let \(N_0\) and \(N_1\) be the two corresponding sets of
variables.

The start symbol belongs to \(N_1\). In a production with its
left side in \(N_1\), exactly one item on the right side produces
the separator: either that item is the terminal \(\#\), or it is
a variable in \(N_1\). The other items are letters in \(\Sigma\)
or variables in \(N_0\).

\paragraph{Variables on either side of the separator.}
Fix a parse tree for \(w\#w^R\in K\) with \(w\in D\).
There is a unique path from its root to the leaf labelled \(\#\).
Every variable occurring off this path derives exactly one
terminal word. To prove this, suppose such an occurrence derives
the word \(z\) used in the tree, but can also derive a different
word \(z'\). The occurrence lies entirely on one side of \(\#\).
If it lies on the left, substituting \(z'\) gives a word
\(x\#w^R\in K\) with \(x\ne w\). Indeed, cancelling the unchanged
prefix and suffix on that side would otherwise give \(z=z'\).
The resulting word is not of the form \(v\#v^R\). It also is not in
\(A\#A^R\), since its right half is \(w^R\) and \(w\notin A\).
This contradicts \(K\subseteq U_A\). If the occurrence lies on
the right, the same argument uses the unchanged left half \(w\).

\paragraph{Replacing these variables by fixed words.}
For each \(B\in N_0\), choose one terminal word \(t_B\) derived
from \(B\). In every production with its left side in \(N_1\),
replace each occurrence of \(B\in N_0\) by \(t_B\), and retain
only the variables in \(N_1\). The resulting grammar \(G_0\)
has productions of the forms
\begin{equation}\label{eq:lineargrammar}
 X\longrightarrow uYv
 \quad\hbox{or}\quad
 X\longrightarrow u\#v,
 \qquad u,v\in\Sigma^*,\quad X,Y\in N_1.
\end{equation}
Let its language be \(K_0\). Every derivation in \(G_0\) can be
expanded to a derivation in \(G\), so \(K_0\subseteq K\).
Moreover, the parse tree considered above is preserved: every
variable replaced in that tree derived only the word already
used there. Therefore
\begin{equation}\label{eq:preserved}
 w\in D,\ w\#w^R\in K
 \quad\Longrightarrow\quad w\#w^R\in K_0.
\end{equation}

\paragraph{Recognizing the words to the left of the separator.}
Now construct a finite directed graph with vertices \(N_1\) and one
additional accepting vertex \(f\). A rule \(X\to uYv\) gives an
edge from \(X\) to \(Y\) labelled \(u\). A rule \(X\to u\#v\)
gives an edge from \(X\) to \(f\) labelled \(u\). The initial
vertex is \(S\). Replacing a word-labelled edge by a finite path
of single-letter edges, and allowing edges that consume no
letter for empty labels, gives a finite automaton.

Its language is
\[
 R=\{x\in\Sigma^*:\text{there is }z\in\Sigma^*
                         \text{ with }x\#z\in K_0\}.
\]
To see this, follow the successive productions in
\eqref{eq:lineargrammar}. The words \(u\) are concatenated in
the order of the corresponding edges. Conversely, an accepting
path specifies a derivation in \(G_0\); the words \(v\) from
those productions determine its right half. Thus \(R\) is regular.

If \(w\in D\cap R\), some \(w\#z\) belongs to
\(K_0\subseteq U_A\). Because \(w\notin A\), this word cannot
belong to \(A\#A^R\). It follows that \(z=w^R\), giving
\(w\#w^R\in K\). The reverse implication in
\eqref{eq:agreement} follows from \eqref{eq:preserved}.
\end{proof}

\begin{proposition}\label{prop:reduction}
Suppose \(A\subseteq\Sigma^*\) is context-free, its complement
\(D\) is infinite, and for every regular language \(R\), either
\(D\cap R\) or \(D\setminus R\) is finite. Then the languages
\[
 L=A\#A^R,\qquad U=L\cup\Delta(\Sigma^*)
\]
are context-free, \(U\setminus L\) is infinite, and every
context-free \(K\) with \(L\subseteq K\subseteq U\) has either
\(K\setminus L\) or \(U\setminus K\) finite.
\end{proposition}

\begin{proof}
The operations of reversal, concatenation, and union preserve context-freeness \cite{ref:hu}.
The language \(\Delta(\Sigma^*)\) is generated by the grammar
\[
 T\longrightarrow\#\ \mid\ aTa\qquad(a\in\Sigma).
\]
Thus \(L\) and \(U\) are context-free, and their difference is
\(\Delta(D)\).
For an intermediate language \(K\) satisfying \eqref{eq:bucher}, Lemma~\ref{lem:grammar} gives a regular
\(R\) for which
\begin{equation}\label{eq:differences}
 K\setminus L=\Delta(D\cap R),\qquad
 U\setminus K=\Delta(D\setminus R).
\end{equation}
The map \(w\mapsto w\#w^R\) is injective. The hypothesis on \(D\)
therefore makes its image \(\Delta(D)\) infinite and makes one
of the two sets in \eqref{eq:differences} finite.
\end{proof}

The use of \(A\) on both sides of the separator is essential to
the preceding argument: when a subtree on one side is changed,
the unchanged half indexed by \(w\notin A\) excludes membership in
\(A\#A^R\). Choosing the two words independently also ensures
that \(A\#A^R\) is context-free.

\section{Testing equations with one counter}\label{sec:counter}

We next describe the automata used to recognize incorrect
encodings of computations. A \emph{nondeterministic one-counter
automaton} consists of a finite-state control and a nonnegative
integer counter; see, for example,
\cite{ref:cl}. It can increment the counter, decrement it when
positive, and test whether it is zero. It may also make $\varepsilon$-moves, which consume no input. Such an automaton is a pushdown
automaton with one stack symbol in addition to a distinguished
bottom symbol, so every language it accepts is context-free.
Finite unions of one-counter languages are one-counter languages:
the automaton first chooses nondeterministically which machine to simulate.
Intersection with a regular language is also possible, by keeping
the state of its finite automaton in the finite-state control.

A signed integer can be represented by storing its absolute value
in the counter and its sign in the finite-state control. A fixed
integer can be added or subtracted by finitely many unit operations.
When the absolute value becomes zero, the sign information is
adjusted accordingly.

Let \(a_1,\ldots,a_d\) be distinct letters, where \(d\geq2\).
In a word
\[
 a_1^{z_1}\cdots a_d^{z_d},\qquad z_1,\ldots,z_d>0,
\]
we call each substring \(a_j^{z_j}\) a \emph{block}. Thus a block
consists of consecutive copies of a single letter, and its length
is \(z_j\). The change from one letter to the next identifies
where a block ends.

\begin{lemma}\label{lem:linear}
For fixed integers \(c_0,c_1,\ldots,c_d\), a one-counter automaton
reading \(a_1^{z_1}\cdots a_d^{z_d}\) can determine the sign of the quantity
\[
 c_0+\sum_{j=1}^d c_jz_j.
\]
Consequently it can test any fixed linear equality or inequality
between the block lengths.
\end{lemma}

\begin{proof}
Initialize the signed counter to \(c_0\). For each occurrence
of \(a_j\), add \(c_j\), using a fixed sequence of unit operations.
The counter represents the displayed expression at the end of
the word, so its sign gives the required test.
\end{proof}

The coefficients in this lemma are fixed parts of the automaton,
not values supplied in the input. The same method can count
positions recognizable by finite control: for example, it can
add one whenever a prescribed sequence of blocks ends.

\begin{lemma}\label{lem:boolean}
Let \(\Phi(z_1,\ldots,z_d)\) be a fixed Boolean formula built from
\(h\geq1\) linear comparisons using conjunction, disjunction, and
negation. If \(r\geq h\), a one-counter automaton can evaluate
\(\Phi\) on inputs consisting of \(r\) identical copies of
\(a_1^{z_1}\cdots a_d^{z_d}\).
It can accept when the formula is true, or when it is false,
as required.
\end{lemma}

\begin{proof}
Use the \(j\)-th copy to evaluate the \(j\)-th comparison by
Lemma~\ref{lem:linear}. Store its truth value in finite control
and empty the counter by $\varepsilon$-moves before testing the next
copy. The letters distinguish successive blocks and copies,
and the number of copies is fixed. After the \(h\) tests,
evaluate the Boolean combination of the stored truth values.
Read any unused copies without changing the counter.
\end{proof}

This machine is not required to check that the copies are identical.
In the construction below, a separate machine detects a disagreement
between copies. Taking the union of the corresponding error
languages will suffice; no closure under intersections of arbitrary
context-free languages is needed.

\section{Encoding computations}\label{sec:encoding}

\subsection{The programs to be encoded}

Fix a program with finitely many control states and \(k\geq1\) registers,
each holding a nonnegative integer. Its initial state and register
values are fixed, and it has no external input. One register is
designated as the output.

Each instruction specifies a source state, a condition on the
current register values, an update, and a target state. The
condition is a fixed Boolean formula in linear comparisons.
The update is an affine map with integer coefficients:
\[
 \mathbf{x}'=M\mathbf{x}+\mathbf{c},
 \qquad M\in\Z^{k\times k},\quad \mathbf{c}\in\Z^k.
\]
The instruction may be used when its source state and condition
match the current data and all its new register values are
nonnegative. All assignments in an instruction are simultaneous.
There are finitely many instructions, and several may be enabled
at once. One application of an instruction counts as one step,
regardless of the sizes of the register values.

A \emph{configuration} specifies the control state and all register
values. Assign a different nonnegative integer to each state, and
write a configuration as
\[
 \mathbf{v}=(v_1,\ldots,v_\ell),\qquad \ell=k+1,
\]
where \(v_1\) is the number assigned to the state and the remaining
coordinates are the register values. Let \(\mathbf{v}_{\mathrm{init}}\)
be the initial configuration, and let \(\nu\) be the coordinate of
the output register. A \emph{computation} is a sequence of
configurations starting at \(\mathbf{v}_{\mathrm{init}}\), with each
successive pair related by an instruction.

Some states are designated as halting states and have no outgoing
instructions. To allow a computation to be extended to a prescribed
number of steps, we add an instruction at every halting state that
leaves the state and all registers unchanged. We call these
\emph{idle steps}. Padding a halting computation means appending
idle steps, or equivalently repeating its final configuration.
These are the only steps permitted after a halting state has been
reached.

\subsection{Words representing steps and computations}

Let \(\Sigma=\{a_0,a_1,\ldots,a_{2\ell}\}\), with all letters
distinct. Fix a positive integer \(B\), which will be a common
multiplier of the block lengths. For two configurations
\(\mathbf{v}\) and \(\mathbf{w}\), define the word
\begin{equation}\label{eq:stepword}
 E_B(\mathbf{v},\mathbf{w})
 =a_0^B
 a_1^{B(v_1+1)}\cdots a_\ell^{B(v_\ell+1)}\\
 \cdot
 a_{\ell+1}^{B(w_1+1)}\cdots
 a_{2\ell}^{B(w_\ell+1)}.
\end{equation}
Its first block records \(B\). The next \(\ell\) blocks represent
the old configuration \(\mathbf{v}\), and the last \(\ell\)
represent the new configuration \(\mathbf{w}\). The addition of
one makes every block nonempty, even when a coordinate is zero.
For example, for \(\ell=2\), \(B=3\),
\(\mathbf{v}=(0,2)\), and \(\mathbf{w}=(1,0)\), the word is
\[
 E_3(\mathbf{v},\mathbf{w})
   =a_0^3a_1^3a_2^9a_3^6a_4^3.
\]
Recovering a coordinate means dividing the corresponding block
length by \(B\) and subtracting one.

Choose a fixed integer \(r\geq2\), to be specified in
Section~\ref{sec:lengthtests}. We group \(r\) successive words of
the form \(a_0^+a_1^+\cdots a_{2\ell}^+\) into a \emph{record} where, as usual,
 \(a_j^+\) means a positive number of copies of \(a_j\).
A record representing the step from \(\mathbf{v}\) to
\(\mathbf{w}\) is
$
 E_B(\mathbf{v},\mathbf{w})^r
$.
Thus it contains \(r\) identical copies of the word in
\eqref{eq:stepword}. These repetitions will allow the one-counter
automaton to test several equations concerning the same step.
When an arbitrary input is divided into records, the \(r\)
words in a record may have different block lengths; equality
of the copies will be one of the tests.

For the fixed program \(P\), let \(D_P\) consist of the words
\begin{equation}\label{eq:DP}
 \prod_{i=1}^{b}E_b(\mathbf{v}_{i-1},\mathbf{v}_i)^r
\end{equation}
such that \(b\geq1\), \(\mathbf{v}_0=\mathbf{v}_{\mathrm{init}}\),
each pair \((\mathbf{v}_{i-1},\mathbf{v}_i)\) is a permitted
instruction or idle step, and \(\mathbf{v}_b\) is a halting
configuration with output \(b\). The product denotes concatenation
in increasing order of \(i\).
In particular, the output \(b\), the number of records, and the
multiplier used in every block length are all equal.
A computation that halts with output \(b\) after \(t\leq b\)
steps is represented by first appending \(b-t\) idle steps and
then writing the word in \eqref{eq:DP}.

\begin{theorem}\label{thm:encoding}
For every program \(P\) of the type described above, one can
choose \(r\) and construct a nondeterministic one-counter
automaton for \(\Sigma^*\setminus D_P\).
Every word of \(D_P\) representing output \(b\) belongs to
\begin{equation}\label{eq:pattern}
 (a_0^+a_1^+\cdots a_{2\ell}^+)^{rb},
\end{equation}
and every block length in it is a positive multiple of \(b\).
\end{theorem}

\subsection{Writing the tests in terms of block lengths}\label{sec:lengthtests}

Consider one word with the letter pattern
\(a_0^+a_1^+\cdots a_{2\ell}^+\). Denote the lengths of its
successive blocks by
\[
 B,\quad X_1,\ldots,X_\ell,\quad Y_1,\ldots,Y_\ell.
\]
If it equals \(E_B(\mathbf{v},\mathbf{w})\), then
\[
 v_j=\frac{X_j}{B}-1,\qquad w_j=\frac{Y_j}{B}-1.
\]
Consequently, for \(B>0\), a comparison
\[
 \sum_j\alpha_jv_j+\sum_j\gamma_jw_j\ \mathrel{\diamond}\ c
\]
where \(\diamond\) is one of \(=,<,\leq,>,\geq\),
is equivalent to the following comparison of block lengths:
\begin{equation}\label{eq:lengthcomparison}
 \sum_j\alpha_j(X_j-B)+\sum_j\gamma_j(Y_j-B)
       \ \mathrel{\diamond}\ cB.
\end{equation}
All coefficients are fixed integers. For example, the update
\(w_j=v_s+v_t\) becomes \(Y_j=X_s+X_t-B\).
This substitution is valid even if the quotients initially
represent rational numbers. We will prove separately that the
tests force integer register values.

We use four fixed Boolean formulas in the written block lengths:
\begin{enumerate}
\item
\(\Phi_{\mathrm{nonneg}}\) is the conjunction of
\(X_j\geq B\) and \(Y_j\geq B\), for \(1\leq j\leq\ell\).
These inequalities say that all represented coordinates are
nonnegative.
\item
\(\Phi_{\mathrm{init}}\) is the conjunction of
\[
 X_j=B(v_{{\mathrm{init}},j}+1)\qquad(1\leq j\leq\ell).
\]
It specifies the initial configuration.
\item
\(\Phi_{\mathrm{step}}\) says that one instruction, including a
possible idle step, relates the two configurations.
For an instruction with source state numbered \(p\), target
state numbered \(q\), and register update
\(\mathbf{x}'=M\mathbf{x}+\mathbf{c}\), its part of this formula
consists of the state equations
\[
 X_1=(p+1)B,\qquad Y_1=(q+1)B,
\]
its condition rewritten by \eqref{eq:lengthcomparison}, and the
update equations
\[
 Y_{s+1}=\sum_{t=1}^k M_{st}X_{t+1}
          +\left(1+c_s-\sum_{t=1}^k M_{st}\right)B
          \quad(1\leq s\leq k).
\]
Take the disjunction of these formulas over all instructions.
\item
\(\Phi_{\mathrm{final}}\) says that the new state is halting
and the new output is positive. It is
\[
 \left(\bigvee_{q\text{ a halting state number}}
          Y_1=(q+1)B\right)\ \wedge\ (Y_\nu>B).
\]
\end{enumerate}

Let \(h\) be the largest number of individual linear comparisons
appearing in any one of these four formulas, and choose
\(r=\max\{2,h\}\). All the formulas and this choice of \(r\)
are determined by the finite description of \(P\).
Lemma~\ref{lem:boolean} now lets a one-counter automaton evaluate
any one of the four formulas on a record whose \(r\) copies agree.

\subsection{Conditions for a word to encode a computation}\label{sec:conditions}

The language of words that can be divided into records is regular:
\[
 \mathcal F=
 \bigl((a_0^+a_1^+\cdots a_{2\ell}^+)^r\bigr)^+.
\]
The changes of letters determine the blocks, and counting
successive occurrences of the letter pattern modulo \(r\)
determines the records. Thus every word in \(\mathcal F\)
has a unique such division.

Suppose a word has \(m\) records. In the first occurrence of
the letter pattern in record \(i\), denote the block lengths by
\[
 (B_i,\mathbf{X}_i,\mathbf{Y}_i).
\]
Thus \(\mathbf{X}_i\) contains the \(\ell\) lengths intended to
represent the old configuration, and \(\mathbf{Y}_i\) contains
those intended to represent the new one. Write
\(O_m=Y_{m,\nu}\) for the length representing the final output.
Consider the following conditions.

\begin{enumerate}
\item\label{cond:copies}
Within each record, the \(r\) occurrences of the letter pattern
have the same corresponding block lengths, so they are identical
words.
\item\label{cond:multiplier}
Every block consisting of the letter \(a_0\) has the same length,
denoted by \(B\). In particular, \(B_i=B\) for all \(i\).
\item\label{cond:join}
\(\mathbf{Y}_i=\mathbf{X}_{i+1}\) for \(1\leq i<m\).
The configuration at the end of a step therefore agrees with
the configuration at the start of the next step.
\item\label{cond:formulas}
Every record satisfies \(\Phi_{\mathrm{nonneg}}\) and
\(\Phi_{\mathrm{step}}\). The first satisfies
\(\Phi_{\mathrm{init}}\), and the last satisfies
\(\Phi_{\mathrm{final}}\).
\item\label{cond:counts}
The number of records and the specified block lengths satisfy
the two equations
\begin{equation}\label{eq:counts}
 m=B_1,\qquad \sum_{i=1}^{m}B_i=O_m-B_m.
\end{equation}
\end{enumerate}

\begin{lemma}\label{lem:characterization}
A word belongs to \(D_P\) if and only if it belongs to
\(\mathcal F\) and satisfies
conditions~(\ref{cond:copies})--(\ref{cond:counts}).
\end{lemma}

\begin{proof}
A word in \eqref{eq:DP} has \(b\) records and uses the multiplier
\(b\) throughout. Its final output is \(b\), so the block length
representing that output is \(b(b+1)\). The two equations in
\eqref{eq:counts} are therefore \(b=b\) and
\(b^2=b(b+1)-b\). The other conditions follow directly from
the encoding of the computation.

Conversely, suppose the conditions hold. The block lengths are
positive integers, but we don't yet know that division by \(B\)
will yield integers. The initial equations give
\[
 X_{1,j}=B(v_{{\mathrm{init}},j}+1).
\]
Hence \(X_{1,j}/B-1=v_{\mathrm{init},j}\), so the first represented
configuration has exactly the prescribed integer coordinates.

Proceed by induction through the records. Suppose that the old
configuration in a record has integer register values and is the
configuration reached by the preceding instructions. The formula
\(\Phi_{\mathrm{step}}\) specifies an instruction, possibly an
idle step. Dividing its equations by \(B\) gives the stated
condition and affine update of that instruction. An affine map
with integer coefficients takes integer inputs to integer
outputs. The formula \(\Phi_{\mathrm{nonneg}}\) ensures
that the new register values are nonnegative, and the state
equation specifies a permitted target state. Thus the new
configuration is obtained by a permitted step of the program.
Condition~(\ref{cond:join}) supplies exactly this configuration
as the old configuration in the next record.

The word therefore represents a computation, including any idle
steps, from the prescribed initial configuration to a halting
configuration with positive integer output \(b\). Its final output
block has length \(O_m=B(b+1)\). Since all \(B_i=B\), the second
equation in \eqref{eq:counts} gives \(mB=Bb\), and therefore
\(m=b\). The first gives \(B=m\). Thus
\[
 B=m=b.
\]
The word is exactly the encoding \eqref{eq:DP}. Once a halting
state is reached, only idle steps are possible, so the
computation halted within its first \(b\) steps.
\end{proof}

This induction explains why a separate test of divisibility
by an input-dependent integer \(B\) is not needed. The initial
equations and the instruction equations force every represented
register value to be an integer.

\subsection{Recognizing the words that fail these conditions}

\begin{proof}[Proof of Theorem~\ref{thm:encoding}]
We construct a finite union of languages accepted by one-counter
automata, one for each kind of error.
Words outside \(\mathcal F\) are accepted by a finite automaton.
For the other tests, the input is restricted to \(\mathcal F\).

An automaton can choose a record by scanning complete records
until it nondeterministically decides to begin a test. Its
finite control remembers which of the \(r\) occurrences of the
letter pattern it is reading, and which block within that
occurrence. To test the last record, it guesses a record and
requires the input to end after that record. These choices
therefore require no information about an unread suffix.

\begin{enumerate}
\item
\emph{Unequal copies within a record.}
Choose two consecutive occurrences of the letter pattern in one
record and one block position. Compare the two lengths at that
position by adding one for each letter of the first block and
subtracting one for each letter of the second. Accept if the
result is nonzero. Comparing consecutive copies suffices to
detect any failure of condition~(\ref{cond:copies}).

\item
\emph{Different lengths for two blocks of the letter \(a_0\).}
Choose two consecutive such blocks and compare their lengths in
the same way. All intervening letters are read without changing
the counter. If condition~(\ref{cond:multiplier}) fails, some
consecutive pair has different lengths.

\item
\emph{Different configurations at the joining of two steps.}
Choose consecutive records \(i,i+1\) and a coordinate \(j\).
Compare \(Y_{i,j}\) with \(X_{i+1,j}\), using the first occurrence
of the letter pattern in each record. The number and order of
the intervening blocks are fixed, so finite control identifies
both blocks while preserving the counter.
Accept if their lengths differ.

\item
\emph{An incorrect instruction or configuration.}
Choose a record and test the negation of
\(\Phi_{\mathrm{nonneg}}\) or \(\Phi_{\mathrm{step}}\),
using Lemma~\ref{lem:boolean}. Alternatively, test the negation
of \(\Phi_{\mathrm{init}}\) on the first record or of
\(\Phi_{\mathrm{final}}\) on the last. When the repeated words
agree, these tests accept exactly a failure of the chosen
formula.

\item
\emph{Failure of \(m=B_1\).}
Subtract one for every letter in the very first block of the
input. Add one at the end of each record, including the last.
The final signed counter is \(m-B_1\). Accept if it is nonzero.

\item
\emph{Failure of the second equation in \eqref{eq:counts}.}
For each record, add the length of its first block, namely \(B_i\).
Guess which record is last before reading its first block.
For that block add twice its length instead of once, and then
subtract the length \(O_m\) representing the new output in the
first occurrence of the letter pattern in that record.
Read the remaining blocks without changing the counter, and
require the end of input after the record. The result is
\[
 \sum_{i=1}^{m}B_i+B_m-O_m.
\]
Accept if it is nonzero. This procedure adds or subtracts fixed
amounts per input letter; it does not multiply two unknown
numbers.
\end{enumerate}

Each test uses one counter by Lemmas~\ref{lem:linear}
and~\ref{lem:boolean}. Its finite control can simultaneously
verify the regular format \(\mathcal F\).
The tests of Boolean formulas may give arbitrary answers if
the purported copies differ. Such a word is already outside
\(D_P\) and is accepted by the first test.

No word of \(D_P\) passes any error test. Conversely, a word
outside \(D_P\) either lies outside \(\mathcal F\), fails one
of the equality conditions, or, by
Lemma~\ref{lem:characterization}, fails a formula or an equation
in \eqref{eq:counts}. Hence some test accepts it.
Their finite union is exactly \(\Sigma^*\setminus D_P\), proving
the one-counter assertion. The letter pattern and the multiples
of \(b\) in the theorem follow directly from \eqref{eq:DP}.
\end{proof}

The two equations in \eqref{eq:counts} have separate purposes.
The sum of the \(a_0\)-block lengths forces the number of records
to equal the computed output. Counting the records then forces
their common multiplier \(B\) to equal that same number.
Both conditions can be checked by adding and subtracting
quantities already written as lengths in the input.

\section{Factorial computations and finite automata}\label{sec:factorial}

\subsection{A fixed program for factorials}

The following program has states \(C,R,H\) and four registers
\((i,f,j,s)\). Initially it is in state \(C\), with
\[
 (i,f,j,s)=(1,1,0,0).
\]
The output register is \(f\), and \(H\) is the halting state.
Each row describes one instruction; all assignments are
simultaneous.

\begin{center}
\begin{tabular}{cccl}
\hline
State & Condition & New register tuple & New state\\
\hline
\(C\) & \(i\geq4\) & \((i,f,j,s)\) & \(H\)\\
\(C\) & true & \((i+1,f,0,0)\) & \(R\)\\
\(R\) & \(j<i\) & \((i,f,j+1,s+f)\) & \(R\)\\
\(R\) & \(j=i\) & \((i,s,j,s)\) & \(C\)\\
\hline
\end{tabular}
\end{center}

Whenever the program is in state \(C\), the registers satisfy
\(f=i!\), as proved below. It may continue by computing the next
factorial, or halt if \(i\geq4\). In state \(R\), it multiplies
the previous factorial by the new value of \(i\), using repeated
addition. All conditions and updates are of the type specified
in Section~\ref{sec:encoding}. For padding, we add only the idle
step at \(H\).

\begin{lemma}\label{lem:factorialrun}
The possible halting outputs are exactly \(n!\), for \(n\geq4\).
For each such \(n\), there is a unique computation that first
reaches \(H\) with output \(n!\). It takes
\begin{equation}\label{eq:runtime}
 t_n=1+\sum_{k=2}^n(k+2)=\frac{n(n+5)}2-2
\end{equation}
steps, and \(t_n<n!\).
\end{lemma}

\begin{proof}
Initially, in state \(C\), the values are \(i=1\) and \(f=1!\).
Suppose the program is in \(C\) with \(i=k-1\) and \(f=(k-1)!\),
and chooses to continue. The step to \(R\) sets \(i=k\),
\(j=0\), and \(s=0\), leaving \(f\) unchanged. After \(j\)
applications of the addition instruction,
\[
 0\leq j\leq k,\qquad s=j(k-1)!.
\]
Exactly \(k\) such steps give \(j=k\) and \(s=k!\).
The next instruction returns to \(C\) and sets \(f=k!\).
Induction proves the assertion \(f=i!\) in state \(C\).
The only choice is whether to halt when \(i\geq4\), so exactly
one computation halts at each \(n\geq4\).

Computing \(k!\) from \((k-1)!\) takes one step from \(C\)
to \(R\), \(k\) addition steps, and one step back to \(C\).
The final step from \(C\) to \(H\) accounts for the extra one
in \eqref{eq:runtime}.
For \(n=4\), \(t_4=16<24\). Also,
\[
 t_{n+1}=t_n+n+3,\qquad
 (n+1)!-n!=n \cdot n!\geq n+3\quad(n\geq4).
\]
Induction gives \(t_n<n!\) for every \(n\geq4\).
\end{proof}

Assign state numbers \(C=0\), \(R=1\), and \(H=2\). The five
configuration coordinates are \((q,i,f,j,s)\), so the encoding
alphabet is \(\Sigma=\{a_0,\ldots,a_{10}\}\).
Choose \(r\) by Theorem~\ref{thm:encoding} for this program.
For each \(n\geq4\), append \(n!-t_n\) idle steps to its unique
halting computation and encode the resulting \(n!\) steps using
the multiplier \(n!\). Denote the resulting word by \(w_n\).
Define
\begin{equation}\label{eq:Dfactorial}
 D=D_P=\{w_n:n\geq4\},\qquad A=\Sigma^*\setminus D.
\end{equation}
By Theorem~\ref{thm:encoding}, \(A\) is accepted by a
nondeterministic one-counter automaton and is therefore context-free.
The first block of \(w_n\) consists of \(n!\) copies of \(a_0\).
The words \(w_n\) are therefore distinct, and \(D\) is infinite.

\subsection{Why finite automata eventually give the same answer}

\begin{lemma}\label{lem:finitepower}
Let \(Q\) be a set of \(q\geq1\) elements and put \(e=q!\).
For every map \(\tau:Q\to Q\) and every \(k\geq1\), we have
\[
 \tau^{ke}=\tau^e.
\]
\end{lemma}

\begin{proof}
For each starting point, repeated application of \(\tau\)
enters a cycle after at most \(q-1\) steps. The cycle length
is at most \(q\), and hence divides \(e\).
Both \(e\) and \(ke\) steps reach this cycle, and the difference
between these numbers of steps is divisible by its length.
The two powers therefore agree at every starting point.
\end{proof}

\begin{theorem}\label{thm:regular}
For the language \(D\) in \eqref{eq:Dfactorial} and every
regular \(R\subseteq\Sigma^*\), either \(D\cap R\) or
\(D\setminus R\) is finite.
\end{theorem}

\begin{proof}
Choose a deterministic finite automaton for \(R\), with \(q\)
states, and let \(e=q!\). Reading the letter \(a_j\) defines
a map \(\tau_j\) on its state set: \(\tau_j(s)\) is the state
reached from \(s\) after that letter.

If \(n\geq\max\{4,q\}\), every block length in \(w_n\) is a
positive multiple of \(n!\), and therefore a positive multiple of \(e\).
By Lemma~\ref{lem:finitepower}, reading any block of the letter
\(a_j\) has exactly the same effect as applying \(\tau_j^e\).
Let \(g\) be the map on the state set obtained by reading
\[
 a_0^e a_1^e\cdots a_{10}^e.
\]
Each word \(E_{n!}(\mathbf{v},\mathbf{w})\) in the encoding of
\(w_n\) has this same effect \(g\), whatever the represented
configurations.
There are \(rn!\) such words in \(w_n\), so reading all of it
has the effect
\[
 g^{rn!}=g^e,
\]
again by Lemma~\ref{lem:finitepower}.
This map is independent of \(n\) once \(n\geq\max\{4,q\}\).
Starting at the initial state, the automaton therefore accepts
all those \(w_n\), or rejects all of them.
Only finitely many words of \(D\) remain, proving the assertion.
\end{proof}

\section{Proof of the main theorem}\label{sec:mainproof}

\begin{proof}[Proof of Theorem~\ref{thm:main}]
Take \(D\) and \(A\) from \eqref{eq:Dfactorial}, and define
\(L,U\) by \eqref{eq:endpoints}.
The complement \(A\) is context-free, \(D\) is infinite, and
Theorem~\ref{thm:regular} gives the property required by
Proposition~\ref{prop:reduction}.
Thus \(L,U\) have all the properties in
Theorem~\ref{thm:main}, except that they are currently languages
over \(\Gamma=\Sigma\cup\{\#\}\), an alphabet of twelve letters.

To obtain binary languages, map \(a_j\), for \(0\leq j\leq10\),
to the four-bit binary expansion of \(j\), including leading
zeros, and map \(\#\) to \(1011\). Extend the map by concatenation
to a homomorphism
\[
 h:\Gamma^*\longrightarrow\{0,1\}^*.
\]
All codewords are distinct and have the same length, so \(h\)
is injective. Closure under homomorphism makes \(h(L)\) and
\(h(U)\) context-free. Their difference is infinite because
\(h(U)\setminus h(L)=h(U\setminus L)\).

Suppose a context-free \(K'\) satisfied
\[
 h(L)\subseteq K'\subseteq h(U)
\]
with both \(K'\setminus h(L)\) and \(h(U)\setminus K'\)
infinite. Context-free languages are closed under inverse
homomorphism, so \(K=h^{-1}(K')\) would be context-free and
satisfy \(L\subseteq K\subseteq U\).
Since every word of \(K'\) belongs to \(h(\Gamma^*)\),
injectivity gives bijections
\[
 K\setminus L\longleftrightarrow K'\setminus h(L),
 \qquad
 U\setminus K\longleftrightarrow h(U)\setminus K'.
\]
Both \(K\setminus L\) and \(U\setminus K\) would then be
infinite, contrary to Proposition~\ref{prop:reduction}.
Consequently \(h(L)\) and \(h(U)\) prove
the desired result.
\end{proof}

\section{Declaration of AI usage}

Most of the ideas in this paper were obtained by the LLM GPT 6 Astra.  Originally the first author (RM) obtained a proof with this LLM, which was then modified repeatedly by Astra, following the suggestions of the second author (JS).  Some of the text was then rewritten by the second author.  The authors take full responsibility for all claims.

\end{document}